\documentclass[preprint]{imsart}
\RequirePackage[OT1]{fontenc}
\RequirePackage{amsthm,amsmath}
\RequirePackage[colorlinks,citecolor=blue,urlcolor=blue]{hyperref}
\usepackage{algorithmic, algorithm}
\usepackage{multirow} 
\usepackage{multicol} 
\usepackage{graphicx}
\usepackage{subfigure} 
\usepackage[authoryear]{natbib}

\startlocaldefs
\numberwithin{equation}{section}
\theoremstyle{plain}
\newtheorem{theorem}{Theorem}[]
\newtheorem{definition}{Definition}[]
\newtheorem{lemma}{Lemma}[]

\endlocaldefs

\begin{document}

\begin{frontmatter}
\title{Variational Bayes  inference and  Dirichlet process priors}
\runtitle{Variational Bayes  inference and  Dirichlet process priors}

\author{\fnms{Hui} \snm{Zhao}\ead[label=e1]{h6zhao@uwaterloo.ca}}
\address{Department of Statistics and Actuarial Science,    University of Waterloo, 200 University Avenue West, Waterloo, Ontario, Canada N2L 3G1 \printead{e1}}

\and
\author{\fnms{Paul } \snm{Marriott}\ead[label=e2]{pmarriot@uwaterloo.ca}}
\address{Department of Statistics and Actuarial Science,    University of Waterloo, 200 University Avenue West, Waterloo, Ontario, Canada N2L 3G1 \\\printead{e2}}

\runauthor{Hui Zhao and Paul Marriott}

\begin{abstract}
This paper shows how the  variational Bayes method provides a computational efficient  technique in the context of hierarchical  modelling using  Dirichlet  process priors, in particular without  requiring  conjugate prior assumption.  It shows, using the so called parameter separation parameterization, a simple criterion under which the variational method works well.     Based on this  framework,  its provides  a full variational solution for the Dirichlet process.  The numerical results show that the  method is very  computationally efficient when compared to MCMC.   Finally,  we propose an empirical method to estimate the truncation level for the truncated Dirichlet  process. 
\end{abstract}


\end{frontmatter}

\section{Introduction}

This article  shows how to apply the variational Bayes (VB) method to hierarchical models which use the  Dirichlet process (DP)  prior.  It shows how the VB method can handle  non-conjugacy in its prior specification, which extends to standard approach to these models.  We also provide a VB approximation to the posterior predictive distribution and compare it with results derived from  two Markov chain Monte Carlo (MCMC)  methods.  For the truncated DP, we  propose an empirical method to determine the number of distinct components in a finite dimensional DP.

In Bayesian parametric modelling,  the prior distribution is usually constructed by assuming it has a particular parametric form.  In many ways, though, it is  more appealing  that the support of the prior is the class of all distribution functions.  In particular, this allows greater flexibility for modelling and inference.   The Dirichlet process, introduced by Ferguson  \citep{ferguson1973bayesian}, provides  a means of specifying  a  probability measure $P(dF)$ over the  space of all (discrete) probability measures. Following this, the  DP has become  very popular when applied to  Bayesian non-parametric inference.  Mixture models are among the important applications of the DP, for example, \cite{escobar1994estimating} and \cite{escobar1995bayesian}. In particular the  clustering property exhibited by the generalized Polya urn representation \cite{blackwell1973ferguson} makes the DP  a natural choice for the prior distribution in the mixture model.

 Markov chain Monte Carlo (MCMC) methods,   in the context of  a DP prior,  have been extensively studied, for example, see  \cite{escobar1994estimating}, \cite{escobar1995bayesian},  \cite{west1993hierarchical}, and \cite{maceachern1994estimating}.  A common aspect of these methods is that they integrate over the random probability measures  and  use the generalized Polya urn representation of the DP.  The Polya urn samplers are restricted to using conjugate base distributions that allow analytic evaluation of the transition probabilities.  When non-conjugate priors are used, these methods require an often difficult numerical integration.   MacEachem and M\"uller \citep{maceachern1998estimating}, and Neal \citep{neal2000markov} devised approaches for handling non-conjugacy by using a set of auxiliary parameters.  

The truncated stick-breaking representation of the DP has also  been considered.  For example    \cite{ishwaran2000markov} shows that with a moderate truncation, the  finite dimensional DP should be able to achieve an accurate approximation.  Based on this representation,  Ishwaran and James \cite{ishwaran2001gibbs} proposed a Gibbs sampler to handle non-conjugacy issue.   

In recent years,  variational Bayesian inference has been applied to  DP based problems, for example see \cite{blei2006variational}.  Strictly  speaking, they used the mean-field method  rather than  a full variational solution, where the approximating distributional family is specified, and the optimization is only over the variational parameters.  In addition, they also only consider the case where the  conjugate base distribution is an  exponential family.

The hierarchical principle is a natural way to model  dependence amongst model parameters.  This article considers an simple, but important,  model based on the normal distribution, in which the observed data are normally distributed with different means for each  group or  experiment, and a normal population distribution is assumed  for the group means. This model is often called the  one-way random-effects model and is widely applicable, being an important special case of the hierarchical linear model.  As \cite{maceachern1994estimating} pointed out,  restricting the prior to be a normal distribution severely constrains the estimate of normal means, producing estimators that shrink each data value toward the same point.   Replacing the normal prior  by a Dirichlet process has been considered by  \cite{maceachern1994estimating} and \cite{bush1996semiparametric} in an MCMC context.   

This article considers non-conjugate settings for this model and presents a full variational Bayesian solution, where the optimization is in terms of both the distributional family and the parameters of the approximating distribution.   The core ingredient for the proposed solution lies on  a special parameterization for a parametric family, called the parameter separation parameterization. This parameterization exhibits some particular algebraic properties, for which the VB approximations possess particularly attractive properties.   In our solution,  we use a truncated stick-breaking representation of the DP.   A natural question is raised by given a dataset  how to estimate the truncation level for a  finite dimensional DP.  We  propose an empirical method to determine the number of distinct components in a finite dimensional DP.

The posterior predictive distribution for this model is not available in a closed form.  For the VB method,   even though we can obtain closed-formed posterior approximations and use them to replace the unknown posterior densities in computing  the posterior predictive density,  it is still not available in a closed form.  In the present paper,  we show how to use the similar variational method to approximate this quantity.  

The rest of the paper is organized as follows. Section 2 presents the one-way random-effects model with a Dirichlet process prior, and shows how to use Gibbs samplers to simulate samples from the posterior distributions.  Section 3  introduces the parameter separation parameterization and a variational  approach on it.  By using these results,  we  obtain the VB solution for the one-way random-effects model with Dirichlet process prior.  Section 4 discusses how to approximate the posterior predictive distributions by the MCMCM methods and by the VB method. Numerical studies are presented in Section 5.   Conclusions are given  in Section 6.

\section{The one-way random effects model}

In this section, we describe the one-way random-effects model which uses a DP prior in a non-conjugate setting, and then show how we can adapt two MCMC methods introduced by \cite{neal2000markov} and \cite{ishwaran2001gibbs} to  obtain the posterior samples. 

In the one-way random effects model, we consider $J$ independent experiments, with experiment $j$ estimating the parameter  $\theta_{j}$ from $n_{j}$ independent normally distributed data points, $y_{ij}$, with a common unknown error variance $\sigma^{2}$. We define $y_{j}$ as $y_{j}=(y_{1j},\cdots, y_{n_{j}j})$. Parameters $\theta_{j}$ are assumed independently drawn from a normal distribution with mean $\mu$ and variance $\tau^{2}$. The parameters of $\mu$, $\tau^{2}$ and $\sigma^{2}$ are  further treated as random variables.  This model is given by 
{\setlength\arraycolsep{0.05em}
\begin{eqnarray} 
y_{ij} | \theta_{j}, \sigma^{2} &\sim& N(\theta_{j}, \sigma^{2}), \nonumber \\ 
\theta_{j}|\mu,\tau^{2} &\sim& N(\mu, \tau^{2}), \nonumber \\
(\sigma^{2}, \mu,\tau^{2}) &\sim& \pi\;\; \text{for } i= 1,\dots, n_{j};  j= 1,\dots, J,  \label{eq:ch3-one way random-effects model}
\end{eqnarray} 
}
where $\pi$ is a prior distribution.  When the normal distribution at the middle stage is  replaced by a DP,  this gives the following model:
\begin{eqnarray} 
y_{ij} | \theta_{j}, \sigma^{2} &\sim& N(\theta_{j}, \sigma^{2}), \nonumber \\ 
\theta_{j} | F &\sim& F, \nonumber \\
F | \alpha, F_{0} &\sim& \text{DP}(\alpha, F_{0}), \nonumber \\
\sigma^{2}&\sim& \pi \;\;\text{ for } i= 1,\dots, n_{j};  j= 1,\dots, J,  \label{eq:ch3-one way random-effects model with DP}
\end{eqnarray} 
where $\alpha$ is a positive real-valued concentration parameter and $F_{0}$ is a base distribution. We consider $F_{0}$ a normal distribution with mean $\mu$ and variance $\tau^{2}$, both  are  further treated as random variables.  It is worth  noting that in this setting $F_{0}$ is not  conjugate to the likelihood. 

The realizations of the  DP are discrete with probability one, thus the above model can be viewed as  a countably infinite mixture \citep{ferguson1983bayesian}.  When  integrating over $F$ in (\ref{eq:ch3-one way random-effects model with DP}), we can  obtain a representation, referred as the generalized Polya urn scheme, of the prior distribution of $\theta_{j}$ in terms of successive conditional distributions of the following form \cite{blackwell1973ferguson}: 
\[ \theta_{j}|\theta_{1},\cdots,\theta_{j-1}= \left\{ 
  \begin{array}{l l}
    \theta_{l} & \; \text{with probability }  \frac{1}{\alpha+j-1} \text{ for each } l \in \{1,\cdots, j-1\}\\
    \sim F_{0} & \; \text{with probability } \frac{\alpha}{\alpha+j-1}
  \end{array} \right. \]
This representation gives a clear view for the clustering or mixture effects of the DP prior, and constitutes a fundamental ingredient for the Polya urn form of  MCMC samplers. 

Alternatively,   \cite{sethuraman1994constructive} provides a constructive definition of the random distribution $F$ in the  DP:
\begin{eqnarray}
F = \sum_{j=1}^{T} v_{j}\delta_{\theta_{j}}, \nonumber 
\end{eqnarray} 
where  $w_{i} \stackrel{\mathrm{iid}}{\sim} \mbox{Beta}(1,\alpha)$, and $v_{j}$ is defined as $v_{1} = w_{1}$,  $v_{j} = w_{j}\prod_{l=1}^{j-1}(1-w_{l})$, and $\theta_{j} \stackrel{\mathrm{iid}}{\sim} F_{0}$, and $\delta_{\theta_{j}}$ denotes a discrete measure concentrated at $\theta_{j}$  ,   and $1 \leq T \leq \infty$. This is often referred to as  the ``stick-breaking'' representation.  If $T < \infty$, this is referred to as a truncated DP or finite dimensional DP \cite{ishwaran2000markov}.   

The exact computation of posterior quantities using  model (\ref{eq:ch3-one way random-effects model with DP}) is typically infeasible. However,  MCMC provides one means to approximate them.  Due to the non-conjugate property of model (\ref{eq:ch3-one way random-effects model with DP}),  we consider using the methods introduced in \cite{neal2000markov} and \cite{ishwaran2001gibbs} to  obtain posterior samples.  

First,  we consider the method proposed by \cite{neal2000markov} and similar to the  ``no gaps'' algorithm proposed earlier by \cite{maceachern1998estimating}.   Let $\zeta = (\zeta_{1},\cdots,\zeta_{K})$ denote the set of distinct $\theta_{j}$, where $j=1,\cdots,J$ and $K \leq J$. Let $c=(c_{1},\cdots,c_{J})$ denote a vector of indicators defined by $c_{j} = k$ if and only if $\theta_{j} = \zeta_{k}$. The state of the Markov chain consist of  $c$, $\zeta$,  $\mu$, $\tau^{2}$ and $\sigma^{2}$.  Each sampling scan consists of  picking a new value for each $c_{j}$ from its conditional distribution given y,  $\zeta$, and all the $c_{l}$ for $ l\neq j $ (written as $c_{-j}$), and then picking a new value for each $\zeta_{k}$ from its conditional distribution given $y$ and $c$, and then picking a new value for $\mu$, $\tau^{2}$ and $\sigma^{2}$ respectively from their conditional distributions.  

\begin{algorithm}[!htp] 
\caption{Polya-urn-type Gibbs sampler}
\label{algm:neal}
\begin{algorithmic}                    
\STATE 
Step 1. For $j=1,\cdots, J$, generate $c_{j}^{(t)}$ from the distribution of $c_{j}|y, \zeta, c_{-j}, \mu, \tau^{2}, \sigma^{2}$.\\
\begin{itemize}
\item Let $k^{-}$ be the number of distinct $c_{l}$ for $l\neq j$, and let $p=k^{-}+s$.  Label $c_{l}$ with values in $\{1,\cdots, k^{-}\}$.
\item Draw values independently from $F_{0}(\mu^{(t-1)}, {\tau^{2}}^{(t-1)})$ for all the $\zeta_{a}^{(t)}$ for which $k^{-}+1\leq a \leq p$. If the value of  $c_{j}^{(t-1)}$ is a singleton (only associated with one $y_{j}$),  then $\zeta_{k^{-}}$ equals to $\zeta_{c_{j}}^{(t-1)}$, otherwise draw a new value for $\zeta_{k^{-}}$ from $F_{0}(\mu^{(t-1)}, {\tau^{2}}^{(t-1)})$.  
\item Draw a value for $c_{j}^{(t)}$ from $\{1,\cdots, p\}$ with the following probability 
\[ P(c_{j}=a|c_{-j}^{(t-1)},y,  {\sigma^{2}}^{(t-1)})\propto \left\{ 
  \begin{array}{l l}
   m_{-j,a}f(y_{j}; \zeta_{a}^{(t-1)}, {\sigma^{2}}^{(t-1)} ), \text{ for } 1\leq a \leq k^{-}\\
   \frac{\alpha}{s}f(y_{j}; \zeta_{a}^{(t)}, {\sigma^{2}}^{(t-1)} ), \text{ for } k^{-} < a \leq p\\
  \end{array} \right. \]
where $m_{-j,a}$ is the number of $c_{l}$ for $l\neq j$ that are equal to c.
\item Discard the $\zeta_{a}$'s that are not now associated with any observation, and relabel $\zeta_{k}$ and corresponding $c_{j}$. 
\end{itemize}

Step 2.  For $k=1,\cdots, |c|$, generate $\zeta_{k}^{(t)}$ from the distribution of $\zeta_{k}|y,  \mu, \tau^{2}, \sigma^{2}$, which is give by 
{\setlength\arraycolsep{0.005em}
 \begin{eqnarray*}
p(\zeta_{k}|y,  \mu^{(t-1)}, {\tau^{2}}^{(t-1)}, {\sigma^{2}}^{(t-1)}) \propto \prod_{j: c_{j} = k} \prod_{i=1}^{n_{j}} \phi(y_{ij}; \zeta_{k}, {\sigma^{2}}^{(t-1)})\phi(\zeta_{k}; \mu^{(t-1)},{\tau^{2}}^{(t-1)}),
 \end{eqnarray*}
 }
where $\phi(.)$ denotes the normal density function. \\

Step 3. Generate $\mu^{(t)}$, ${\tau^{2}}^{(t)}$, and ${\sigma^{2}}^{(t)}$  from the corresponding full conditional distribution, that are given as follows:
{\setlength\arraycolsep{0.005em}
 \begin{eqnarray*}
p({\sigma^{2}}| y, \zeta_{k}^{(t)})  &\propto& \prod_{k=1}^{|c|} \prod_{j: c_{j} = k} \prod_{i=1}^{n_{j}} \phi(y_{ij}; \zeta_{k}^{(t)}, {\sigma^{2}}) \pi(\sigma^{2}) \\
p(\mu |y,  \zeta_{k}^{(t)}, {\tau^{2}}^{(t)}) &\propto& \prod_{k=1}^{|c|} \phi(\zeta_{k}^{(t)}; \mu, {\tau^{2}}^{(t)})\pi(\mu) \\
p(\tau^{2} |y, \zeta_{k}^{(t)}, {\mu}^{(t)}) &\propto& \prod_{k=1}^{|c|} \phi(\zeta_{k}^{(t)}; \mu^{(t)}, {\tau^{2}})\pi(\tau^{2}),
 \end{eqnarray*}
where $\pi(\tau^{2})$, $\pi(\sigma^{2})$, and $\pi(\mu)$ are corresponding priors.  
}
\end{algorithmic}
\end{algorithm}

The key feature to handle the issue of non-conjugacy lies that when $c_{j}$ is updated, a set of size of $s$ temporary auxiliary parameter variables that represent possible values for $\zeta_{k}$ that are not associated with any other observations is introduced.  Since the observations $y_{j}$ are exchangeable, we can assume that we are updating $c_{j}$ for the last observation, and that the $c_{l}$ for other observations have values in the set $\{1, \cdots, k^{-}\}$, where $k^{-}$ is the number of distinct $c_{l}$  for $l \neq j$. By using the auxiliary variables, the possible values for a new $c_{j}$ lies in $\{1, \cdots, k^{-}, k^{-}+1, \cdots, k^{-}+s \}$.  Once a new value for $c_{j}$ has been chosen,  all the $\zeta$ that are not now associated with any observation will be discarded, and $\zeta_{k}$ and the corresponding $c_{j}$ are  relabeled to have the $c_{j}$ with values in $\{1,\cdots, |c|\}$, where $|c|$ denotes the number of distinct number in $c$.  This Gibbs updating  for model (\ref{eq:ch3-one way random-effects model with DP}) is summarized in Algorithm \ref{algm:neal}.

In addition to handling the issue of non-conjugacy,  \cite{neal2000markov}  suggests that this method can improve the mixing of the chain and shorten the autocorrelation time  to reduce the sample size used to estimate the posterior quantities.  However, it is clear that since $F$ is integrated over,  this Polya-urn like sampler still  restricts the inference for the posterior of the random $F$ to be based only on the posterior for  $\zeta_{k}$'s, that is, there  no  explicit inference on $F$ is possible. The paper \cite{ishwaran2001gibbs}  devised a, so called, blocked Gibbs sampler, which uses the stick-breaking representation, to avoid the limitation imposed by the Polya urn like samplers.  

The  key to the blocked Gibbs sampler lies that it is infeasible to work on an infinite numbers of components in the stick-breaking representation,  and it has to truncate the DP at a curtain level, denoted as $B$, and discard the components of $B+1$, $B+2$, $\cdots$. \cite{ishwaran2001gibbs}  shows that  with a moderate truncation the marginal density under a truncated DP prior is indistinguishable from the one based on the infinite DP prior.  By using a stick-breaking representation, the one-way random-effects model given in (\ref{eq:ch3-one way random-effects model with DP}) under a truncated DP can be written as follows:
{\setlength\arraycolsep{0.05em}
\begin{eqnarray} 
&&y_{ij}|c_{j}, \zeta_{}, \sigma^{2} \sim N(\zeta_{c_{j}}, \sigma^{2}), \;\; \text{for } i= 1,\dots, n_{j};  j= 1,\dots, J, \nonumber\\
&& c_{j}|v \sim  \sum_{b=1}^{B} v_{b}\delta_{b}; \;\; v_{1} = w_{1},  v_{b} = w_{b}\prod_{l=1}^{b-1}(1-w_{l}), \nonumber\\
&& w_{b} \sim \mbox{Beta}(1,\alpha), \text{ for } b=1,\cdots, B-1, \text{ and }  w_{B} =,1 \nonumber\\
&&\zeta_{b} \sim N(\mu, \tau^{2});\; \text{for } b= 1,\dots, B,  \;\; \nonumber\\
&&(\sigma^{2}, \mu,\tau^{2})\sim \pi.   \label{eq:ch3-one way random-effects model with DP with SB}
\end{eqnarray} 
}
In this model, the  state of the Markov chain consist of  $c$, $\zeta$, $v$, $\mu$, $\tau^{2}$ and $\sigma^{2}$.  The blocked Gibbs sampling for  model  (\ref{eq:ch3-one way random-effects model with DP with SB}) is summarized in Algorithm \ref{algm:ishwaran}.
\begin{algorithm}[!htp] 
\caption{Blocked Gibbs sampler}
\label{algm:ishwaran}
\begin{algorithmic}                    
\STATE 
Step 1. For $j=1,\cdots, J$, generate $c_{j}^{(t)}$ from the distribution of $c_{j}|y, \zeta, v, \sigma^{2}$, that is given by:
{\setlength\arraycolsep{0.005em}
 \begin{eqnarray*}
 p(c_{j}|y, &&\zeta, v, \sigma^{2})  = \sum_{b=1}^{B} p_{b, j}\delta_{b}, \; \text{ where } p_{b,j} \propto v_{b}^{(t-1)}\prod_{i=1}^{n_{j}}\phi(y_{ij};\zeta_{b}^{(t-1)}, \sigma^{2}{^{(t-1)}})
 \end{eqnarray*}
}

Step 2.  For $b=1,\cdots, B$, generate $\zeta_{b}^{(t)}$ as follows:  
\begin{itemize}
\item When  $\zeta_{b}^{(t)}$ is not associated with any $y_{j}$,  draw a new value from  $F_{0}(\mu^{(t-1)}, {\tau^{2}}^{(t-1)})$.
\item Otherwise,  draw a new value from  the following conditional distribution:
{\setlength\arraycolsep{0.005em}
 \begin{eqnarray*}
p(\zeta_{b}|y,  \mu^{(t-1)}, {\tau^{2}}^{(t-1)}, {\sigma^{2}}^{(t-1)}) \propto  \prod_{j: c_{j} = b} \prod_{i=1}^{n_{j}} \phi(y_{ij}; \zeta_{b}, {\sigma^{2}}^{(t-1)})\phi(\zeta_{b}; \mu^{(t-1)},{\tau^{2}}^{(t-1)}),
 \end{eqnarray*}
}

\end{itemize}

Step 3. Generate $v^{(t)}$  from the following conditional distribution:
{\setlength\arraycolsep{0.005em}
 \begin{eqnarray*}
v_{1}^{(t)} &=& w_{1}^{(t)},  v_{b}^{(t)} = w_{b}^{(t)}\prod_{l=1}^{b-1}(1-w_{l}^{(t)}),\\
w_{b}^{(t)} &\sim& \mbox{Beta}(M_{b},\alpha+\sum_{l=b+1}^{B}M_{l}); \;\text{$M_{b}$ is the number of $c_{j}$ equals to $b$}
 \end{eqnarray*}
}

Step 4. Generate $\mu^{(t)}$, ${\tau^{2}}^{(t)}$, and ${\sigma^{2}}^{(t)}$  from the corresponding full conditional distribution, that are given as follows:
{\setlength\arraycolsep{0.005em}
 \begin{eqnarray*}
p({\sigma^{2}}| y, \zeta_{b}^{(t)})  &\propto& \prod_{b=1}^{B} \prod_{j: c_{j} = b} \prod_{i=1}^{n_{j}} \phi(y_{ij}; \zeta_{b}^{(t)}, {\sigma^{2}}) \pi(\sigma^{2}) \\
p(\mu |y,  \zeta_{b}^{(t)}, {\tau^{2}}^{(t)}) &\propto& \prod_{b=1}^{B} \phi(\zeta_{b}^{(t)}; \mu, {\tau^{2}}^{(t)})\pi(\mu) \\
p(\tau^{2} |y, \zeta_{b}^{(t)}, {\mu}^{(t)}) &\propto& \prod_{b=1}^{B} \phi(\zeta_{b}^{(t)}; \mu^{(t)}, {\tau^{2}})\pi(\tau^{2}),
 \end{eqnarray*}
where $\pi(\tau^{2})$, $\pi(\sigma^{2})$, and $\pi(\mu)$ are corresponding priors.  
}
\end{algorithmic}
\end{algorithm}

\section{Variational Bayesian method}

As an alternative to MCMC methods, the VB method provides analytical approximations to  posterior quantities and in practice it  has been demonstrated to be  very much faster to implement.  The core of the method builds on the basis of maximization of a lower bound of the  logarithm of the marginal likelihood.  Early developments of the method can be found in the applications on neural networks, \cite{1993-Hinton&Camp}, and \cite{1995-Mackay}. The method has been successfully applied in many different disciplines and domains, and in recent years, it has obtained more attention from both the application and theoretical perspective in the mainstream of Statistics, for example, see \cite{2002-Hall-Humphreys-Titterington}, \cite{2006-Wang-Titterington}, \cite{2010-Ormerod-Wand}, \cite{2009-McGrory-Titterington-Reeves-Pettitt}, and \cite{2011-Faes-Ormerod-Wand}.

Generally, variational inference has been mainly developed in the context of the exponential family.  For example,  \cite{beal2003variational} and \cite{wainwright2008graphical} provide a general variational formalism for the conjugate exponential family.  There  are several limitations with these developments.  First,  they mainly consider the cases assuming   conjugate priors.  Second,  the variational inferences are developed only with respect to natural parameters, which are often not the parameters of immediate interests. In the present paper,  we show that  VB inferences can be extended to a more general situation, where we consider a particular form of a parameterization for a parametric family, which  we call the {\em parameter separation parameterization}, which is defined as follows:
\begin{definition}\label{df: LPSF}
A parametric family $\{P_{\tau}: \tau \in R^{d}\}$ is said to  have a parameter separation parameterization if and only if the logarithm of its { density function} can be written as 
\begin{eqnarray}
\log f(y) = h(y)+\displaystyle \sum_{c=1}^{C} \left( \prod_{i=1}^{d} g_{c,i}(\tau_{i}, y) \right), \label{eq:defition-LPSF}
\end{eqnarray} 
where $C$ is a positive integer, and $h$ and $g_{c,i}$ are real-valued functions.
\end{definition}
Many distributions can be written in the form of (\ref{eq:defition-LPSF}).  We can list a few examples:   normal,  inverse Gamma, Pareto, Laplace, Weibull, finite discrete distributions. These include both exponential family and non-exponential family examples.  An important feature of this representation lies that when taking expectation on $\log f(y)$,  the  right hand side of (\ref{eq:defition-LPSF}) provides a factorized form, which is the key to make possible the construction of  the analytical form of the variational distributions.   Moreover, we will see  from the following theorem that with this parameterization the distributional families of VB approximations have particularly tractable  forms and these forms are not changed during the iterative updates. Also  the convergence of variational parameters can be used as the stopping rule for the iterative updates of the VB method instead to evaluate the computationally burdensome lower bound. 

VB gains its computational advantages by making simplifying assumptions about the posterior dependence structure.  A full factorization, which assumes that  all  model parameters are independent of each other in the approximating distribution,  is the  most commonly used scheme. However, we consider a factorization scheme in which more flexible dependence structures can be used.  Suppose  $\tau$ is a $d$ dimension parameter vector, indexed by  $I = \{1, \cdots, d\}$. We consider a VB approximation for the posterior $p(\tau|y)$, which is factorized as,  
\begin{eqnarray}
q(\tau) = \displaystyle \prod_{i}^{K}q(\tau_{\mathcal{F}_{i}}) = \displaystyle \prod_{i}^{K} q(\tau_{\mathcal{C}_{i}} |\tau_{\mathcal{P}_{i}})  q(\tau_{\mathcal{P}_{i}}), \label{eq:ch2-VBMA-factorization}
\end{eqnarray}
where $\{\mathcal{F}_{i}\}_{i=1}^{K}$ is a partition of the index set $I$, for $K \leq d$, and $\mathcal{F}_{i} = \mathcal{C}_{i} \cup \mathcal{P}_{i}$ and $\mathcal{C}_{i} \neq \emptyset$ for $i=1,\cdots, K$.  If the set $\mathcal{P}_{i}$ is  an empty set, then $q(\tau_{\mathcal{C}_{i}} |\tau_{\mathcal{P}_{i}})$ denotes the  unconditional  density $q(\tau_{\mathcal{C}_{i}})$.   We denote $\setminus \mathcal{C}$ as the complement set of $\mathcal{C}$  in $I$.

The following theorem gives a general formularization for the variational inference on the parameter separation parameterization with a factorization scheme given in (\ref{eq:ch2-VBMA-factorization}).
\begin{theorem}
\label{thm: ch2-DP-VBMA}
Suppose $y = \{y_{j}\}_{j=1}^{J}$ are $i.i.d.$ from a distribution having a parameter separation parameterization,  where $\tau \in R^{d}$, then 
\begin{description}
\item[($i$)]  the $q(\tau_{\mathcal{C}_{i}} | \tau_{\mathcal{P}_{i}} )$ in  (\ref{eq:ch2-VBMA-factorization}) is given by
{\setlength\arraycolsep{0.005em}
\begin{eqnarray}
&&q(\tau_{\mathcal{C}_{i}} | \tau_{\mathcal{P}_{i}} )  
 \propto  \pi(\tau_{\mathcal{C}_{i}} | \tau_{\mathcal{P}_{i}}) \exp (\displaystyle \sum_{j}^{J}  ( h(y_{j})+\displaystyle \sum_{c=1}^{C_{C_{i}}} g_{c}(\tau_{\mathcal{C}_{i}}, y_{j}) K_{C_{i},c}^{*} + J_{C_{i}}^{*} ) ), \label{eq: thm-ch2-DP-VBMA-1}
\end{eqnarray}
} 
where $K_{C_{i},c}^{*} = E_{ q( \mbox{\boldmath{$\tau$} }_{\setminus \mathcal{C}_{i}\cup\mathcal{P}_{i}  } ) }\left [K_{C_{i},c}\right ] $ and $J_{C_{i}}^{*}$= $E_{ q( \mbox{\boldmath{$\tau$} }_{\setminus \mathcal{C}_{i}\cup\mathcal{P}_{i}  } ) }\left [J_{C_{i}} \right ]$, $K_{{C}_{i},c}$ and $J_{{C}_{i}}$ are constant with respect to  $\tau_{\mathcal{C}_{i}}$, and $\pi(\tau_{\mathcal{C}_{i}} | \tau_{\mathcal{P}_{i}})$ is a prior.

\item[($ii$)]   the  $q( \tau_{\mathcal{P}_{i}} )$ in  (\ref{eq:ch2-VBMA-factorization}) is given by 
{\setlength\arraycolsep{0.005em}
\begin{eqnarray}
q(\tau_{\mathcal{P}_{i}} ) \propto&&  \pi(\tau_{\mathcal{P}_{i}} )\exp\left( E_{ q(\tau_{\mathcal{C}_{i}} | \tau_{\mathcal{P}_{i}}  )   }\bigg[\log \frac{\pi(\tau_{\mathcal{C}_{i}} | \tau_{\mathcal{P}_{i}} )}{q(\tau_{\mathcal{C}_{i}} | \tau_{\mathcal{P}_{i}}  ) } \bigg] \right)\nonumber \\
&&  \exp (\displaystyle \sum_{j}^{J} ( h(y_{j})+\displaystyle \sum_{c=1}^{C_{P_{i}}} g_{c}(\tau_{\mathcal{P}_{i}}, y_{j}) K_{P_{i},c}^{*} + J_{P_{i}}^{*} ) ) 
\label{eq: thm-ch2-DP-VBMA-2}
\end{eqnarray}
}
where $K_{P_{i},c}^{*} = E_{ q( \tau_{\setminus \mathcal{P}_{i}}) }\left [K_{P_{i},c}\right ] $ and $J_{P_{i}}^{*}=E_{ q( \tau_{\setminus \mathcal{P}_{i}}) }\left [J_{P_{i}} \right ]$, $K_{{P}_{i},c}$ and $J_{{P}_{i}}$ are constant with respect to  $\tau_{\mathcal{P}_{i}}$, and $\pi( \tau_{\mathcal{P}_{i}})$ is a prior. 
\end{description}
\end{theorem}
\begin{proof}: see Appendix. \end{proof}
Given that the density function of $y_{j}$ can be expressed as (\ref{eq:defition-LPSF}),  we can write the likelihood function with respect to $\tau_{\mathcal{C}_{i}}$  as follows:
{\setlength\arraycolsep{0.005em}
\begin{eqnarray}
p(y|\tau)  &=&  \exp (\displaystyle \sum_{j}^{J}  ( h(y_{j})+\displaystyle \sum_{c=1}^{C_{C_{i}}} g_{c}(\tau_{\mathcal{C}_{i}}, y_{j}) K_{C_{i},c} + J_{C_{i}})), \label{eq: ch2-DP-VBMA-likelihood-1} 
\end{eqnarray}
}
where $K_{{C}_{i},c}$ and $J_{{C}_{i}}$ are given in (\ref{eq: thm-ch2-DP-VBMA-1}). We see that the expression  (\ref{eq: thm-ch2-DP-VBMA-1})   shares the same set of functions of $\{g_{c}(\tau_{\mathcal{F}_{i}}, y_{i})\}_{c=1}^{C_{i}}$ with (\ref{eq: ch2-DP-VBMA-likelihood-1}), and the difference between (\ref{eq: thm-ch2-DP-VBMA-1}) and (\ref{eq: ch2-DP-VBMA-likelihood-1})  only lies on  the constant terms of $\{K_{C_{i},c}^{*}\}_{c=1}^{C_{C_{i}}}$ and $J_{C_{i}}^{*}$ up to the prior.  It is similar for $q( \tau_{\mathcal{P}_{i}})$ in (\ref{eq: thm-ch2-DP-VBMA-2}).  This implies that given the likelihood function,  the distributional forms of (\ref{eq: thm-ch2-DP-VBMA-1}) and (\ref{eq: thm-ch2-DP-VBMA-2}) are fixed, and then the lower bound of the log marginal likelihood becomes a function of the parameters of approximation distributions.  The convergence of these parameters is sufficient to guarantee the convergence of the lower bound.  Due to the linearity property of expectation,  Theorem \ref{thm: ch2-DP-VBMA} is easy to be extended to a hierarchical setting, as long as at each layer or stage,  the parametric family has a parameter separation parameterization.

Theorem \ref{thm: ch2-DP-VBMA} is ready to be used in developing the variational inference for the one-way random-effects model with the DP prior. Here, we consider the stick-breaking representation given in (\ref{eq:ch3-one way random-effects model with DP with SB}).  We define $c_{j}$ in (\ref{eq:ch3-one way random-effects model with DP with SB}) as $c_{j}= (c_{j1},\cdots,c_{jB})$, where $c_{jb}$ is an indicator variable with probability $v_{b}$  of  equalling  to one. This probability is  given in  (\ref{eq:ch3-one way random-effects model with DP with SB}). The joint probability of $y, c, v, \zeta, \sigma^{2}, \mu, \tau^{2}$ is given as follows:
{\setlength\arraycolsep{0.005em}
\begin{eqnarray}
p(y, c, v, \zeta, \sigma^{2}, \mu, \tau^{2}) = &&\nonumber\\
&&\prod_{j=1}^{J}\prod_{b=1}^{B}\left\{v_{b}\prod_{i=1}^{n_{j}} \phi(y_{ij};\zeta_{b},\sigma^{2})\right\}^{c_{jb}} \prod_{b=1}^{B}\phi(\zeta_{b};\mu,\tau^{2}) \nonumber\\
&&\prod_{b=1}^{B-1}\text{Beta}(w_{b};1, \alpha) \pi(\sigma^{2})\pi(\mu)\pi(\tau^{2}), \label{eq: joint with DP}
\end{eqnarray} }
where $\pi(\sigma^{2})$,$\pi(\mu)$, and $\pi(\tau^{2})$ are the prior distributions.  To have these priors  providing little influence on the posterior distributions,  we assign non-informative uniform priors for $\mu$, $\log(\sigma^{2})$, and $\tau^{2}$. If we were to assign a uniform prior distribution for $\log(\tau^{2})$, the posterior distribution would be improper.  Thus,  we get the prior distribution for $\mu$, $\log(\sigma^{2})$, and $\tau^{2}$ is given by $\pi(\sigma^{2}, \mu, \tau^{2}) \propto \frac{1}{\sigma^{2}}$

We denote $q(c, v, \zeta, \sigma^{2}, \mu, \tau^{2})$ as the VB approximation for the posterior distribution of $p(c, v, \zeta, \sigma^{2}, \mu, \tau^{2}|y)$.  In contrast to the mean field approximation,  we do not require any distributional families to $q$, except for the  independence assumption.  We assume $q$ has the following factorization form:
{\setlength\arraycolsep{0.005em}
\begin{eqnarray}
&&q(c, v, \zeta, \sigma^{2}, \mu, \tau^{2})= \prod_{j=1}^{J}q(c_{j})\prod_{b=1}^{B}q(v_{b}) \prod_{b=1}^{B}q(\zeta_{b}) q(\sigma^{2})q(\mu|\tau^{2})q(\tau^{2}).
\end{eqnarray}
}
It is worth noting  that using  a full factorization with $q(\mu, \tau^{2})= q(\mu)q(\tau^{2})$,   results in that the convergence of variational parameters fails in the iterative updates. 

It is straightforward to check that the distributions at each stage of model (\ref{eq:ch3-one way random-effects model with DP with SB}) all have a parameter separate parameterization, and then Theorem \ref{thm: ch2-DP-VBMA} is can be used.  By plugging (\ref{eq: joint with DP}) into (\ref{eq: thm-ch2-DP-VBMA-1}) or (\ref{eq: thm-ch2-DP-VBMA-2}), we can obtain the following results:
{\setlength\arraycolsep{0.005em}   
\begin{eqnarray}
q(c_{j}) &=& \text{ Multinomial}(r_{j1},\cdots, r_{jB}) \nonumber\\
r_{jb}  &\propto& \exp\{-\frac{1}{2}\frac{g}{h}\sum_{i=1}^{n_{j}}(y_{ij}-a_{b})^{2} -\frac{1}{2}\frac{g}{h}b_{b}^{2}n_{j}+ \psi(c_{b})\nonumber\\
&& - \psi(c_{b}+d_{b}) +\sum_{l=1}^{b-1}(\psi(c_{l}) - \psi(c_{l}+d_{l})) \}\nonumber\\
q(\zeta_{b}) &=& N(a_{b}, b_{b}^{2}); \;\; a_{b} = \frac{\frac{g}{h}\sum_{j=1}^{J}r_{jb}(\sum_{i=1}^{n_{j}}y_{ij})+\frac{k}{s}e}{\frac{g}{h}\sum_{j=1}^{J}r_{jb}n_{j}+\frac{k}{s}}, b_{b}^{2} = \frac{1}{\frac{g}{h}\sum_{j=1}^{J}r_{jb}n_{j}+\frac{k}{s}}\nonumber\\
q(v_{b}) &=& \text{Beta}(c_{b}, d_{b});\;\;c_{b}= \sum_{j=1}^{J}r_{jb}+1, d_{b}= \sum_{l=b+1}^{B}\sum_{j=1}^{J}r_{jb}+\alpha\; (\text{for } b<B), d_{B}= \alpha;\nonumber\\
q(\mu|\tau^{2}) &=& N(e, \frac{\tau^{2}}{f^{2}});\;\; e=\frac{\sum_{b=1}^{B}a_{b}}{B}, f^{2} = B \nonumber\\
q(\tau^{2}) &=& \text{IG}(k,s);\;\; k=\frac{B}{2}-\frac{3}{2}, s=\frac{1}{2}\sum_{b=1}^{B}((a_{b}-e)^{2}+b_{b}^{2}) \nonumber\\
q(\sigma^{2}) &=& \text{IG}(g,h);\;\; g=\frac{\sum_{j=1}^{J}n_{j}}{2}, h=\frac{1}{2}\sum_{j=1}^{J}\sum_{b=1}^{B}r_{jb}(\sum_{i=1}^{n_{j}}(y_{ij}-a_{b})^{2}+b_{b}^{2}), \label{vb-sol}
\end{eqnarray}
}
where $\psi$ denotes the digamma function, and $\text{IG}$ denotes the gamma distribution.  The above approximations are well-recognised distributions, and they are easy to use to make further inference on parameters.  The VB algorithm requires an iterative updates on the parameters of $r_{jb}$, $a_{b}$, $b_{b}^{2}$, $c_{b}$, $ d_{b}$, $e$, $f$, $g$, $h$, $k$, and $s$ till they converge.

\section{The predictive distribution} 

The posterior predictive distribution provides a distribution for a new data point given the observed data, in which it makes use of the entire posterior distribution.     Suppose $y^{*}=(y^{*}_{1},\cdots, y^{*}_{n^{*}})$ is a new observation,  then the posterior predictive distribution of $y^{*}$ given $y$ is defined as
{\setlength\arraycolsep{0.005em}   
\begin{eqnarray}
p(y^{*}|y) = \int p(y^{*}|\Theta)p(\Theta|y) d\Theta, \label{ppd}
\end{eqnarray}} where $\Theta$ refers the model parameters.  For the one-way random-effects model with a DP prior  this quantity is intractable however   MCMC methods provide a straightforward approximation.  Having a sample of $T$ points from the posterior,  we can estimate it by 
{\setlength\arraycolsep{0.005em}   
\begin{eqnarray}
p(y^{*}|y) = \frac{1}{T}\sum_{t=1}^{T} p(y^{*}|\Theta^{(t)}), \label{ppd-approxi}
\end{eqnarray}
}
where $\Theta^{(t)}$ is the sample drawn from the posterior distribution after the chain reaches its stationary distribution.  For Algorithm \ref{algm:neal}, $p(y^{*}|\Theta^{(t)})$ is given as follows:
{\setlength\arraycolsep{0.005em}   
\begin{eqnarray}
p(y^{*}|\Theta^{(t)}) = \sum_{k=1}^{|{c^*}^{(t)}|} P({c^*}^{(t)}=k )f(y^{*}|\zeta_{k}^{(t)}, {\sigma^{2}}^{(t)}) \nonumber
\end{eqnarray}
}
where again $|{c^*}^{(t)}|$ denotes the number of values which ${{c^*}^{(t)}}$ takes.  For Algorithm \ref{algm:ishwaran}, it is given as follows:
{\setlength\arraycolsep{0.005em}   
\begin{eqnarray}
p(y^{*}|\Theta^{(t)}) = \sum_{b=1}^{B}v_{b}^{(t)} f(y^{*}|\zeta_{b}^{(t)}, {\sigma^{2}}^{(t)}) \nonumber
\end{eqnarray}
}

For the VB method,  it is  natural  to use the VB approximations to replace the unknown posterior distributions in (\ref{ppd}).  Thus, we can have the following approximation for the posterior predictive distribution:
{\setlength\arraycolsep{0.005em}   
\begin{eqnarray}
&&p(y^{*}|y) \approx \int \left( \sum_{b=1}^{B}v_{b} f(y^{*}|\zeta_{b}, {\sigma^{2}}) \right)d Q(v, \zeta, \sigma^{2}) \nonumber\\
 &=& \sum_{b=1}^{B} \text{E}_{q(v_{b})}[v_{b}] \int\left(f(y^{*}|\zeta_{b}, {\sigma^{2}}) \right)d Q(\zeta_{b})d Q(\sigma^{2})
 \label{ppd-vb}
\end{eqnarray}
}
where $Q$ is the VB approximation. Unfortunately, although we have obtained the simple and well-recognised distributions for $Q(\zeta_{b})$ and $Q(\sigma^{2})$,  the integrals in (\ref{ppd-vb}) are still not available in a closed form. However, we can apply the variational principle again to obtain a lower bounds on this quantity, and propose using this lower bound as an approximation for  the posterior predictive distribution.
  
We denote $L_{b}$ as $L_{b} = \int\left(f(y^{*}|\zeta_{b}, {\sigma^{2}}) \right)d Q(\zeta_{b})d Q(\sigma^{2})$.  If we regard $Q(\zeta_{b})$ and $Q(\sigma^{2})$ as prior  distributions,  then $L_{b}$ can be regarded as a marginal likelihood, that can be approximated by the  variational method.  We denote $v(\zeta_{b})$ and $v(\sigma^{2})$ as the variational approximations which result  from  treating $Q(\zeta_{b})$ and $Q(\sigma^{2})$ as priors. Again,  Theorem \ref{thm: ch2-DP-VBMA} can be used to obtain the distributional forms for $v(\zeta_{b})$ and $v(\sigma^{2})$, and gives the following results:
{\setlength\arraycolsep{0.005em}   
\begin{eqnarray*}
v(\zeta_{b}) &=& N(A_{b}, B_{b}^{2}); \nonumber\\
&& A_{b} = \frac{\frac{G}{H}\sum_{i=1}^{n^{*}}y_{i}^{*}+\frac{a_{b}}{b_{b}^{2}}}{\frac{G}{H}n^{*}+\frac{1}{b_{b}^{2}}}, B_{b}^{2} = \frac{1}{\frac{G}{H}n^{*}+\frac{1}{b_{b}^{2}}}\nonumber\\
v(\sigma^{2}) &=& \text{IG}(G,H);\nonumber\\
&& G=g+\frac{n^{*}}{2}, H=h+\frac{1}{2}S^{*} + \frac{n^{*}}{2}((A_{b}-\bar{y^{*}})^{2}+B_{b}^{2}), 
\end{eqnarray*}
}
where $n^{*}$ is the number of observations in $y^{*}$, and $\bar{y^{*}}$ is the mean of $y^{*}$, and $S^{*}$ is the total sum of squares of $y^{*}$, and $a_{b}$, $b_{b}^{2}$, $g$, and $h$ are given in (\ref{vb-sol}). 

Once the variational parameters of $A_{b}$, $B_{b}^{2}$, $G$, and $H$ converge,  we can obtain a lower bound of the logarithm of $L_{b}$, denoted as $F_{b}$, which is given as follows:
{\setlength\arraycolsep{0.005em}   
\begin{eqnarray*}
F_{b}&=& \int  \frac{q(\zeta_{b})}{v(\zeta_{b})} dV(\zeta_{b}) + \int  \frac{q(\sigma^{2})}{v(\sigma^{2})} dV(\sigma^{2}) + \int \log(f(y^{*}|\zeta_{b}, {\sigma^{2}})) dV(\zeta_{b}, {\sigma^{2}}) \\
&=& \log(\frac{1}{b_{b}}) -\log(\frac{1}{B_{b}})-\frac{1}{2b_{b}^{2}}((A_{b}-a_{b})^{2}+B_{b}^{2}) \\
&&+(G-g)(\log H - \psi(G))+G(1-\frac{h}{H})+\log\frac{h^{g}}{\Gamma(g)}+\log\frac{H^{G}}{\Gamma(G)}\\
&&-\frac{n^{*}}{2}(\log 2\pi+ \log H - \psi(G))-\frac{1}{2}\frac{G}{H}(\sum_{i=1}^{n^{*}}(y_{i}^{*}-A_{b})^{2}- n^{*}B_{b}^{2}),
\end{eqnarray*}
}

where $\Gamma(.)$ is the gamma function. 

Once we obtain the values of each $F_{b}$ for $b=1,\cdots, B$,  we can obtain a lower bound for (\ref{ppd-vb})
{\setlength\arraycolsep{0.005em}   
\begin{eqnarray*}
\sum_{b=1}^{B} \text{E}_{q(v_{b})}[v_{b}] L_{b} \geq \sum_{b=1}^{B} \text{E}_{q(v_{b})}[v_{b}] \exp(F_{b}) \equiv F.
\end{eqnarray*} 
}
Thus,  we propose to use $F$ as an approximation for the posterior predictive distribution of $p(y^{*}|y)$.

\section{Numerical studies} 

We examine the performance of the VB method by comparing it with the two MCMC methods on simulated data.  To generate the data, we set $\mu$ and $\tau^{2}$ for the base distribution in (\ref{eq:ch3-one way random-effects model with DP}) to be $\mu = 0$ and $\tau^{2} = 16$ and $\sigma^{2}$ equal to $0.64$.  We use the truncated stick-breaking representation to construct the random distribution $F$.  For  demonstration purposes,  we simply truncate $F$ at level 5, shown in Table \ref{tbl:random distribution $F$}.  A data set of 60 groups data are generated from $F$ , and each group contains 80 data points. We use 50 groups as the observed data and  10 groups as the future data. 
\begin{table}[h]
\caption{A random distribution $F$, truncated at level 5}
\label{tbl:random distribution $F$}
\begin{center}
\begin{tabular}{ c | c   c   c   c   c}
\cline{1-6}
 $\zeta_{b}$  &  -2.22 &  -0.54   & 1.01  &  4.28    & 7.10\\ 
$P(\zeta_{b})$ &  0.35 & 0.14    & 0.13  &  0.13    & 0.26\\ 
\end{tabular}
\end{center}
\end{table}

In the VB learning,  we assume we have  no knowledge about the distribution $F$, and also mis-specify the truncation level to 10.   The algorithm converges after 19 iterations. 
\begin{table}[h]
\caption{The VB approximations  for the random distribution $F$ }\label{tbl:VB mean}
\begin{center}
\begin{tabular}{ c c c c c c c c c c }
\hline 
   E$[v_{1}]$ &   E$[v_{2}]$&   E$[v_{3}]$&   E$[v_{4}]$ &   E$[v_{5}]$ &  E$[v_{6}]$ &   E$[v_{7}]$&   E$[v_{8}]$&   E$[v_{9}]$ &   E$[v_{10}]$\\
 E$[\zeta_{1}]$ &   E$[\zeta_{2}]$&   E$[\zeta_{3}]$&   E$[\zeta_{4}]$ &   E$[\zeta_{5}]$ &  E$[\zeta_{6}]$ &   E$[\zeta_{7}]$&   E$[\zeta_{8}]$&   E$[\zeta_{9}]$ &   E$[\zeta_{10}]$\\
\hline 
 0.167 & 0.16 & 0.12  & 0.12 & 0.01  & 0.01 & 0.01  & 0.13   & 0.13  & 0.11 \\
 -2.24 &  -2.24 & -0.55  & 0.97  &   2.06 & 2.06  & 2.06  & 4.23   & 7.12   & 7.12
\end{tabular}
\end{center}
\end{table}
Table \ref{tbl:VB mean} gives the expected values for $v_{b}$ and $\zeta_{b}$ under the VB approximations. We can see a clear pattern. The expected probability weights for the component 5, 6, and 7, are close to zero.  This may suggest they can be ruled out from the true model.  The component 1 and 2 share the exact same value of $-2.24$, which is close to the value of component 1 in Table \ref{tbl:random distribution $F$}, and the cumulated expected probability weight of $0.327$ is also close to $0.35$ in Table \ref{tbl:random distribution $F$}.  We can observe a similar situation for component 9 and 10.  Thus,  by combining same components (with same values) and ruling out the empty components (with very small probability weights),  we can conclude that VB picks up 5 components for the random distribution $F$.  
  
For the Polya-urn type Gibbs sampler (Algorithm \ref{algm:neal}), we run $2\times10^{5}$ iterations. The computational time for this Gibbs sampler is about more than 10,000 times of the one required by the VB method. We use the last $20\%$ data, which we believe the chain has reached its stationary distribution. To reduce the serial correlation effect, we pick the every $25^{th}$ data point.   The frequencies of the distinct number of $\zeta_{b}$ are given in Table \ref{tbl:Polya-urn like Gibbs number of N_active}.  We see that the posterior probability  favors 5, 6, or 7 components, and 6 components has the largest probability.   
\begin{table}[h]
\caption{Posterior probabilities for the number of $\zeta$}
\label{tbl:Polya-urn like Gibbs number of N_active}
\center
\begin{tabular}{ p{1.2cm} | p{0.6cm}  p{0.6cm}p{0.6cm}p{0.6cm}p{0.6cm}p{0.6cm} }
\cline{1-7}
 Num.  &  5 &  6   & 7  &  8    & 9 & 10\\ 
P(Num.) &  0.270  & 0.386     & 0.254   &  0.068     & 0.018 & 0.002 
\end{tabular}
\end{table}
For the blocked Gibbs sampler (Algorithm \ref{algm:ishwaran}), we run $2.5 \times10^{6}$ iterations, which requires about 15,000 times of the computational time as mush as the VB method. The last $20\%$ data is used. To reduce the serial correlation effect, we pick the every $25^{th}$ data point.  Even with the order constraints on $\zeta$,  the chain still shows the signs of label switching.   Thus, a single value of $v_{b}$ or $\zeta_{b}$ may lose the interpretability.

Finally, we compare the posterior predictive distribution approximated by the three methods.  We compute   the log predictive likelihoods, shown in Table \ref{tbl:pllh},  for the 10 groups of future data.  For the Gibbs samplers, additional 2,500 samples are collected and used in the computation.  We see that the three methods  give very close values.  The mean values are given as $-95.95$, $-97.30$, $-97.32$ respectively.  A $t$ test, for the log predictive likelihoods computing by Algorithm \ref{algm:ishwaran} and by VB, is performed, and it can not reject the hypothesis that the true difference in means is equal to 0 at a p-value equal to 0.9923, and we also can obtain a p-value equal to 0.5049 for Algorithm \ref{algm:ishwaran} versus  Algorithm \ref{algm:neal},
\begin{table}[h]
\caption{Log predictive likelihood for 10 groups of future data}\label{tbl:pllh}
\center
\begin{tabular}{ r r r }
Polya-urn & Blocked & VB\\
\hline 
-96.19  &  -97.40  & -97.29  \\
-98.43  & -99.67 & -99.88 \\ 
-89.45   &  -90.59  &  -90.46 \\
-97.35  & -98.53  &  -98.74 \\
-104.31  & -105.84  & -105.90 \\
-95.64  & -96.76  & -96.88  \\
-90.36  & -91.50   & -91.37 \\
-99.84  &-100.82      & -100.62  \\
-92.86 & -95.53   &-95.47  \\
-95.11 & -96.32  & -96.54 
\end{tabular}
\end{table}
\section{Discussion}
The variational Bayes method provides a computational efficient  technique to approximate certain posterior quantities in the context of hierarchical  modelling using  Dirichlet  process priors.  To avoid the limitation in the existing variational formalism which relies on conjugate exponential families, we consider VB in a new framework. The parameter separation parameterization gives  a factorization which allows  flexible dependence structures.    Based on this new framework,  we provide a full variational solution for the Dirichlet process with non-conjugate base prior.  The numerical results show that the VB method is  very computationally efficient.  Moreover, the comparison with two different MCMC methods shows that VB provides accurate approximations for the posterior predictive distribution.  Finally,  we propose an empirical method to estimate the truncation level for the truncated DP.

\section{Appendix} 

To prove Theorem \ref{thm: ch2-DP-VBMA}, we give the following lemma first.

\begin{lemma}
\label{thm: ch2-VBMA}
Let $p(y,\tau)$ be the  joint distribution of data  $y$  and a model parameter vector $\tau$. 
The VB approximations $of q(\tau_{\mathcal{C}_{i}}|\tau_{\mathcal{P}_{i}})$ and $q( \tau_{\mathcal{P}_{i}})$ in (\ref{eq:ch2-VBMA-factorization})  are given by
{\setlength\arraycolsep{0.005em}   
\begin{eqnarray}
q(\tau_{\mathcal{C}_{i}} | \tau_{\mathcal{P}_{i}})  &\propto&  \exp\left( E_{ q( \tau_{\setminus (\mathcal{C}_{i} \cup \mathcal{P}_{i})} ) }[\log  p(y, \tau) ]  \right), \label{eq:thm-vb-1} \\
q( \tau_{\mathcal{P}_{i}} )  &\propto&  \exp \left(- E_{ q(\tau_{\mathcal{C}_{i}} | \tau_{\mathcal{P}_{i}} )  }[\log  q(\tau_{\mathcal{C}_{i}} | \tau_{\mathcal{P}_{i}} ) ] \right) \exp \left( E_{ q( \tau_{\setminus\mathcal{P}_{i} } )}[\log  p(y, \tau) ]  \right). 
\label{eq:thm-vb-2}
\end{eqnarray}
}
\begin{proof}: The Kullback-Leible divergence from $q(\tau)$ to $p(\tau|y)$ can be written as
{\setlength\arraycolsep{0.005em}   
\begin{eqnarray}
KL(q(\tau)||p(\tau|y) )= \log p(y)  - \int q(\tau) \log   \frac{p(\tau , y) } {q(\tau)} d\tau. \label{eq: ch2-KL2LB} 
\end{eqnarray}
}
Plugging (\ref{eq:ch2-VBMA-factorization}) into (\ref{eq: ch2-KL2LB}) and re-arrange the terms with respect to $q(\tau_{\mathcal{C}_{i}} | \tau_{\mathcal{P}_{i}})$, we can obtain the following expression:
{\setlength\arraycolsep{0.005em}   
\begin{eqnarray}
&&\text{KL}(q(\tau)||p(\tau|y) )=\nonumber\\
&& E_{q(\tau_{\mathcal{P}_{i}} )} [\mbox{KL} ( q(\tau_{\mathcal{C}_{i}} | \tau_{\mathcal{P}_{i}} ) ||\frac{1}{Z} \exp( E_{ q( \tau_{\setminus\mathcal{C}_{i}\cup\mathcal{P}_{i} ) } }[\log  p(y, \tau) ]  )   ) ]+ \log p(y) +  K \label{eq:ch2-thm-KL-C}
\end{eqnarray}
}
where $Z$ is a normalization constant,   and K is a constant with respect to $q(\tau_{\mathcal{C}_{i}} | \tau_{\mathcal{P}_{i}})$. The first term on the right hand side of  (\ref{eq:ch2-thm-KL-C}) is the only term which depends  on $q(\tau_{\mathcal{C}_{i}} | \tau_{\mathcal{P}_{i}})$.  Then,  the minimum value of $\text{KL}[q(\tau) || p(\tau|y) ]$ is achieved when  the first term of the right-hand side of  (\ref{eq:ch2-thm-KL-C}) equals to zero.  Thus,  we obtained 
{\setlength\arraycolsep{0.005em} 
 \begin{eqnarray}
q(\tau_{\mathcal{C}_{i}} | \tau_{\mathcal{P}_{i}}) = \frac{1}{Z}\exp\left( E_{ q( \tau_{\setminus (\mathcal{C}_{i} \cup \mathcal{P}_{i})} ) }[\log  p(y, \tau) ]  \right). \nonumber 
\end{eqnarray}
}
Similar to (\ref{eq:thm-vb-2}).
\end{proof}
\end{lemma}

\subsection{Proof of Theorem \ref{thm: ch2-DP-VBMA}} 

We write the  joint distribution of $p(y, \tau)$ as $p(y|\tau)\pi(\tau)$, where  $\pi(\tau)$ is a prior distribution.  Given that the density function of $y_{j}$ can be written in the form of (\ref{eq:defition-LPSF}), we can write the likelihood function with respect to $\tau_{\mathcal{C}_{i}}$ as
{\setlength\arraycolsep{0.005em} 
\begin{eqnarray}
p(y| \tau)&& =\exp (\displaystyle \sum_{j}^{J}  ( h(y_{j})+\displaystyle \sum_{c=1}^{C_{{C}_{i}}} g_{c}(\tau_{\mathcal{C}_{i}}, y_{j}) K_{{C}_{i},c} + J_{{C}_{i}} ) ), \label{eq: ch2-DP-VBMA-likelihood-1-app} 
\end{eqnarray}
}

where $K_{{C}_{i},c}$ and $J_{{C}_{i}}$ are constant with respect to  $\tau_{\mathcal{C}_{i}}$.  We assume that the priors have the following forms
{\setlength\arraycolsep{0.005em} 
\begin{eqnarray}
\pi(\tau) = \displaystyle \prod_{i}^{K} \pi(\tau_{\mathcal{C}_{i}} |\tau_{\mathcal{P}_{i}})  \pi(\tau_{\mathcal{C}_{i}}), \label{eq:ch2-VBMA-factorization-priors}
\end{eqnarray}
}
Thus, $p(y, \tau)$ in (\ref{eq:thm-vb-1} ) can be replaced by (\ref{eq: ch2-DP-VBMA-likelihood-1-app}) and (\ref{eq:ch2-VBMA-factorization-priors}), and then the results of (\ref{eq: thm-ch2-DP-VBMA-1}) is a direct application of the linearity of expectation. Similarly, we can obtain the result of (\ref{eq: thm-ch2-DP-VBMA-2}).

\pagebreak

\bibliographystyle{imsart-nameyear}
\bibliography{vb-dp.bib}

\end{document}